\documentclass{article}
\usepackage{amssymb}

\usepackage{graphicx}
\usepackage{amsmath,mathrsfs}

\newtheorem{theorem}{Theorem}

\newtheorem{corollary}[theorem]{Corollary}

\newtheorem{definition}[theorem]{Definition}

\newtheorem{lemma}[theorem]{Lemma}

\newtheorem{proposition}[theorem]{Proposition}
\newtheorem{remark}[theorem]{Remark}

\newenvironment{proof}[1][Proof]{\textbf{#1.} }{\ \rule{0.5em}{0.5em}}

\begin{document}

\title{Asymptotic Equivalence of Quantum Stochastic Models}
\author{Luc Bouten$^{1}$, John E. Gough $^{1,2}$}
\date{}
\maketitle

\hyphenation{mo-dels}

\begin{abstract}
We introduce the notion of perturbations of quantum stochastic models using the series product, and establish the asymptotic convergence of sequences of quantum stochastic models under the assumption that they are related via a right series product perturbation. While the perturbing models converge to the trivial model, we allow that the individual sequences may be divergent corresponding to large model parameter regimes that frequently occur in physical applications. This allows us to introduce the concept of asymptotically equivalent models, and we provide several examples where we replace one sequence of models with an equivalent one tailored to capture specific features. \\
\indent
These examples include: a series product formulation of the principle of virtual work; essential commutativity of the noise in strong squeezing models; the decoupling of polarization channels in scattering by Faraday rotation driven by a strong laser field; and an application to quantum local asymptotic normality.
\end{abstract}

1) Visiting scholar, Institut Henri Poincar\'{e}, 11 rue Pierre et Marie Curie, 
75231 Paris Cedex 05, France.

2) Institute for Mathematics, Physics and Computer Science, Aberystwyth
University, Aberystwyth, Ceredigion, SY23 3BZ, Wales

\section{Introduction}

The Hudson-Parthasarathy theory of quantum stochastic evolutions with Fock-space noise \cite{HP,Par92}, has been applied to various open quantum models. Quantum stochastic differential equations (QSDEs) arise via a weak coupling procedure \cite{AFL90,G05} and are widely used in the field of quantum optics. Together with the quantum input-output theory of Gardiner-Collett \cite{Gardiner_Collett,GarZol00}, the quantum filtering theory of V.P. Belavkin \cite{Belavkin92,BvHJ07}, and the results on adiabatic elimination \cite{GvH07,BS08}, this lays the bedrock for a powerful new theory of quantum feedback networks \cite{GouJam09a,GouJam09b} that allows for a description of interconnected quantum components communicating via quantum fields.

In many physical applications, specific parameters may be very large or very small. This motivates a study of the limit behavior of quantum stochastic models in extreme parameter limit regimes. Often we would like to replace the original model with one which is asymptotically equivalent, but where the dominant terms are more clearly defined, or where desired characteristics are more transparent. In this paper, we show that perturbing a (possibly divergent) sequence of models using a \textit{right} series product \cite{GouJam09a,GouJam09b} construction leads to a natural class of asymptotically equivalent models. 

We consider perturbations to the coefficients of unitary quantum stochastic differential equations, arising through the right series product, which converge to the neutral element in the uniform topology. It turns out (Theorem \ref{thm:main_divergent}) that the perturbed and unperturbed models are asymptotically equivalent in the sense that the difference between the unitary processes they generate is asymptotically convergent in the strong topology. We make use of the Trotter-Kato Theorem in the proof, essentially using the graph convergence of the generators to show the uniform convergence of certain semigroups associated to the unitaries. The proof is completed by a density argument on the Fock space.

As applications of our main result, we consider a diverse array of examples. An example arising from a concrete physical problem is the strong squeezing limit in quantum optics: here a desirable feature of the model is that one quadrature of the field dominates and one would like to replace the original model with an equivalent one where only this quadrature appears explicitly. Theorem \ref{thm:main_divergent} makes this replacement precise, and also leads to a new Hamiltonian term indicating the squeezing effect. A second concrete example comes from the Faraday rotation of $y$-polarized light into the $x$-polarization direction by a cold atom gas. Here we show that in the weak rotation and strong driving regime the $x$ and $y$ polarization directions decouple, again by replacing the original sequence of models with an appropriately decoupled asymptotically equivalent sequence.

We also apply Theorem \ref{thm:main_divergent} to more general principles. We elucidate the principle of virtual work for open quantum systems \cite{AFG12} by reformulating the model variation as a \textit{left} series product perturbation, and show that vanishing perturbations lead to convergence of the corresponding evolutions. Finally, we show how quantum local asymptotic normality \cite{CBG} may be treated within our framework. Here we consider a family of open quantum system models parameterized by an unknown parameter $\theta$. We assume a central limit type scaling $\theta = \theta_0 + \frac{v}{\sqrt{n}}$ on the parameter $\theta$, with standard assumptions on the Taylor series expansion of the QSDE coefficients about $\theta_0$. In addition, we take the long time limit $t \hookrightarrow nt$. In the large $n$ limit, we obtain an equivalent asymptotic model based on the truncated Taylor expansion. The series product structure of the perturbation nicely demonstrates the appearance of extra Hamiltonian terms. 

The remainder of this article is organized as follows. In Section \ref{sec:SLH}, we recall the Hudson-Parthasarathy theory of QSDEs \cite{HP} and the related SLH modeling. Here the notion of perturbations via the series product is introduced, along with the notion of asymptotically equivalent models. Moreover, this section formulates the main result, Theorem \ref{thm:main_divergent}. In Section \ref{sec:proof} we prove the Theorem. The examples, representing applications of our main result, are then presented in Section \ref{sec:applications}.

\section{The series product of quantum Markov systems}
\label{sec:SLH}
\subsection{Quantum stochastic process description}

We begin with a description of a quantum Markovian system, colloquially called an \textit{SLH model}, see
Figure 1, which may be considered as a single component in a quantum network. It is prescribed by fixing the Hilbert space of the system $\mathfrak{h} $ the multiplicity $n$ of quantum inputs so that the noise space is is
the Fock space 
\begin{eqnarray*}
\mathfrak{F}=\Gamma \left( L_{\mathfrak{K}}^{2}\left( \mathbb{R},dt\right)
\right) ,
\end{eqnarray*}
where $\Gamma \left( \cdot \right) $ is the bosonic Fock space functor and $\mathfrak{K}=\mathbb{C}^{n}$ is the multiplicity space, and finally the
collection of operators $G \sim (S, L, H) $ with 
\begin{eqnarray*}
S\triangleq \left[ 
\begin{array}{ccc}
S_{11} & \dots & S_{1n} \\ 
\vdots &  & \vdots \\ 
S_{n1} & \cdots & S_{nn}
\end{array}
\right] , \quad L\triangleq \left[ 
\begin{array}{c}
L_{1} \\ 
\vdots \\ 
L_{n}
\end{array}
\right] ,
\end{eqnarray*}
and $H$ self-adjoint.

\begin{center}
\setlength{\unitlength}{.04cm} 
\begin{picture}(120,45)
\label{pic1}
\thicklines
\put(45,10){\line(0,1){20}}
\put(45,10){\line(1,0){30}}
\put(75,10){\line(0,1){20}}
\put(45,30){\line(1,0){30}}
\thinlines
\put(48,20){\vector(-1,0){45}}
\put(120,20){\vector(-1,0){20}}
\put(120,20){\line(-1,0){48}}
\put(50,20){\circle{4}}
\put(70,20){\circle{4}}
\put(100,26){inputs}
\put(48,35){system}
\put(56,17){$G$}
\put(10,26){outputs}
\end{picture}

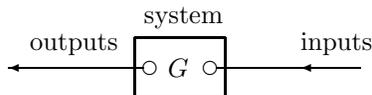
Figure 1: input-output component
\end{center}

With the specification $G\sim (S,L,H)$, we have a
quantum stochastic evolution, $\{ U_{G}(t):t\geq 0\}$ which is the adapted
unitary quantum stochastic process on $\mathfrak{h}\otimes \mathfrak{F}$ in
the sense of Hudson and Parthasarathy \cite{HP} occurring as the solution to
the quantum stochastic differential equation (QSDE) 
\begin{eqnarray}
dU_{G}=\left( dG\right) U_{G},\quad U_{G}\left( 0\right) =1, 
\end{eqnarray}
where 
\begin{eqnarray}
dG\left( t\right) &=&\left( S_{ij}-\delta _{ij}\right) \otimes
d\Lambda _{ij}\left( t\right) +L_{i}\otimes dA_{i}^{\ast  }\left(
t\right)  \nonumber \\
&&-L_{i}^{\ast  }S_{ij}\otimes dA_{j}\left( t\right) +%
K\otimes dt,
\end{eqnarray}
where 
\begin{eqnarray}
K=-\frac{1}{2}L_{i}^{\ast  }L_{i}-iH.
\label{eq:K}
\end{eqnarray}
Note that 
\begin{eqnarray}
K+K^\ast  =L_{i}^\ast  L_{i}.
\label{eq:K+Kdag}
\end{eqnarray}
Here we encounter Ito differentials with respect to the fundamental process
of scattering, creation, annihilation and time on the Fock space. To recall,
we fix the standard orthonormal basis $\left\{ e_{j}:j=1,\cdots ,n\right\} $
for $\mathfrak{K}=\mathbb{C}^{n}$ and take $A_{i}\left( t\right) \triangleq
A\left( e_{i}\otimes 1_{\left[ 0,t\right] }\right) $ and $A_{i}^{\ast 
}\left( t\right) \triangleq A^{\ast  }\left( e_{i}\otimes 1_{\left[ 0,t\right]
}\right) $ to be the operators describing the \textit{annihilation} and 
\textit{creation} of a quantum in the $i$th channel over the time interval $%
\left[ 0,t\right] $, respectively; the operator describing the scattering
from the $j$th channel\ to the $i$th channel over the time interval $\left[
0,t\right] $ is $\Lambda _{ij}\left( t\right) $ which is the differential
second quantization of the tensor product on $\mathfrak{K}\otimes
L^{2}[0,\infty )$ of $|e_{i}\rangle \langle e_{j}|$ tensored with the
operator of multiplication by $1_{\left[ 0,t\right] }$.

Operators on the system evolve according to $j_{t}\left( X \right)
\triangleq U_{G}^{\ast  }\left( t\right) \left[ X\otimes 1\right]
U_{G}\left( t\right) $, and we have 
\begin{eqnarray}
dj_{t}\left( X\right) &=&j_{t}\left( \mathscr{L}\left( X%
\right) \right) dt+j_{t}(S_{ji}^{\ast  }\left[ X,L%
_{j}\right] )dA_{i}^{\ast  }+j_{t}\left( [L_{i}^{\ast  },X]%
S_{ij}\right) dA_{j}  \nonumber \\
&&+j_{t}\left( \sum_{k}S_{ki}^\ast  XS%
_{kj}-\delta _{ij}X\right) d\Lambda _{ij}(t),  \label{eq:HL}
\end{eqnarray}
with the Lindblad generator
\begin{eqnarray}
\mathscr{L}\left( X\right) \triangleq \frac{1}{2}L_{i}^{\ast  }%
\left[ X,L_{i}\right] +\frac{1}{2}\left[ L%
_{i}^{\ast  },X\right] L_{i}-i\left[ X,H%
\right] \equiv L_{i}^{\ast  }XL_{i}-K^{\ast  }%
X-XK.
\label{eq:Lindbladian}
\end{eqnarray}

The output fields are defined by $A_i^{\text{out}} (t) \triangleq U_G^{\ast 
}\left( t\right) \left[ 1\otimes A_{i}\left( t\right) \right] U_G\left(
t\right)$, and we have 
\begin{eqnarray}
dA^{\text{out}}_{i}\left( t\right) &=&j_{t}\left( S_{ij}\right)
dA_{j}\left( t\right) +j_{t}\left( L_{i}\right) dt.
\end{eqnarray}

We can put two such open systems in cascade, see Figure 2, where the output
of $G_1$ is fed in as the input of $G_2$ under the limit where the time of
propagation between the components vanishes. Here we obtain the reduced
quantum Markov component, $G_{\text{series}}$, with 
\begin{eqnarray}
S_{\text{series}} &=&S_{2}S_{1}, \\
L_{\text{series}} &=&L_{2}+S_{2}L_{1}, \\
H_{\text{series}} &=&H_{1}+H_{2}+\text{Im}\left\{ 
L_{2}^{\ast  }S_{2}L_{1}\right\} .
\end{eqnarray}

\begin{center}
\setlength{\unitlength}{.1cm} 
\begin{picture}(100,22)
\label{pic2}

\thicklines

\put(30,5){\line(0,1){10}}
\put(30,5){\line(1,0){20}}
\put(50,5){\line(0,1){10}}
\put(30,15){\line(1,0){20}}

\put(60,5){\line(0,1){10}}
\put(60,5){\line(1,0){20}}
\put(80,5){\line(0,1){10}}
\put(60,15){\line(1,0){20}}

\thinlines
\put(32,10){\vector(-1,0){15}}
\put(62,10){\vector(-1,0){14}}

\put(92,10){\vector(-1,0){14}}

\put(33,10){\circle{2}}
\put(63,10){\circle{2}}

\put(47,10){\circle{2}}
\put(77,10){\circle{2}}

\put(38,9){$G_2$}
\put(68,9){$G_1$}

\end{picture}

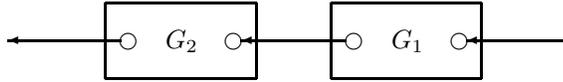
Figure 2: Systems in series/cascade
\end{center}

In the case of two cascaded systems with Hilbert spaces $\mathfrak{h}_i$, ($%
i=1,2)$, then we have the joint Hilbert space $\mathfrak{h} \cong %
\mathfrak{h}_1 \otimes \mathfrak{h}_2$ and we should interpret $G_1 \equiv (%
S_1 \otimes I_2, L_1 \otimes I_2, H_1 \otimes I_2
) $, etc. However, the expression $G_{\text{series }}$ makes sense even if
we do not have this factorization: that is, we cannot separate out the
internal degrees of freedom of the two components and should look at this as
a double-pass through the same single component first via coupling $G_1$ and
then via $G_2$, see Figure 3.

\begin{center}
\setlength{\unitlength}{.04cm} 
\begin{picture}(120,45)
\label{pic3}
\thicklines
\put(45,10){\line(0,1){40}}
\put(45,10){\line(1,0){30}}
\put(75,10){\line(0,1){40}}
\put(45,50){\line(1,0){30}}
\thinlines
\put(48,20){\vector(-1,0){40}}
\put(120,20){\vector(-1,0){20}}
\put(120,20){\line(-1,0){48}}
\put(50,20){\circle{4}}
\put(70,20){\circle{4}}
\put(48,40){\line(-1,0){48}}
\put(48,40){\vector(-1,0){45}}
\put(120,40){\vector(-1,0){20}}
\put(120,40){\line(-1,0){48}}
\put(50,40){\circle{4}}
\put(70,40){\circle{4}}
\put(0,5){\line(1,0){120}}
\put(0,5){\line(0,1){35}}
\put(120,5){\line(0,1){15}}
\put(100,46){inputs}
\put(56,17){$G_2$}
\put(56,37){$G_1$}
\put(10,26){outputs}
\end{picture}

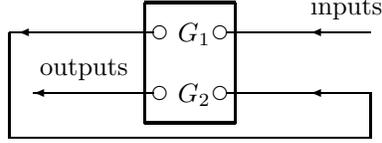
Figure 3: The series product $G_2 \vartriangleleft G_1 $ as double-pass
through the same component.
\end{center}

In this more general setting, we say that we have an associative,
non-commutative product defined, on the set of all models with the same
system Hilbert space, $\mathfrak{h}$, and multiplicity space, $\mathfrak{K}$%
, by 
\begin{eqnarray}
G_{2}\vartriangleleft G_{1}\equiv G_{\text{series}}. 
\end{eqnarray}
We refer to this as the series product, and it is a group law with identity $%
E\sim (I_{\mathfrak{K}},0,0)$. The inverse of the series product is well defined: $G \sim (S,L,H)$ has $G^{-1} \sim
(S^\ast  , -S^\ast   L , -H)$ with $G^{-1} \vartriangleleft G = E = G \vartriangleleft G^{-1}$.

\begin{remark}
An alternative convention is to represent the model as $G\sim \left[ S,L,K\right] $ where the non-selfadjoint damping operator 
$K$ is used instead of $H$. The series product for this
representation takes the form 
\begin{eqnarray}
\left[ S_{2},L_{2},K_{2}\right] \vartriangleleft %
\left[ S_{1},L_{1},K_{1}\right] =\left[ S%
_{2}S_{1},L_{2}+S_{2}L_{1},K%
_{1}+K_{2}-L_{2}^\ast  S_{2}L_{1}\right]
. 
\end{eqnarray}
\end{remark}

\begin{definition}
Given a model $G\sim (S, L, H) $ and $\delta G
\sim (\delta S, \delta L, \delta H) $ another
model with the same system Hilbert space, $\mathfrak{h}$, and multiplicity
space, $\mathfrak{K}$. Then we say that a left/right series product
perturbation of $G$ by $\delta G$ is given by $\delta G \vartriangleleft G$, 
$G \vartriangleleft \delta G$, respectively.
\end{definition}

For a perturbation $\delta G \sim ( e^{i \delta \Theta}, \delta 
L, \delta H) $ we obtain the right series product
perturbation of $G \sim (S, L, H) $ given by 
\begin{eqnarray}
\widetilde G &=&  G \vartriangleleft \delta G  \nonumber \\
&=& \left( S e^{i \delta  \Theta } , \, L + S \delta L, 
\, H+ \delta H + \text{Im}
\left\{  L^\ast   S \delta L \right\} \right) 
\nonumber \\
&\equiv & (\widetilde{S}, \widetilde{L }, \widetilde{%
H }) .  \label{eq:pert_SLH}
\end{eqnarray}

We now state our main result.

\begin{theorem}
\label{thm:main_divergent} 
Let $(G^{(k)})_{k}$ be a sequence of models with the same system Hilbert
space, $\mathfrak{h}$, and multiplicity space, $\mathfrak{K}$. Let $(\delta
G^{(k)})_{k}$ be a sequence of perturbations of the form $\delta G^{(k)}\sim
\left( e^{i\delta \Theta ^{\left( k\right) }},\delta L^{\left( k\right)
},\delta H^{\left( k\right) }\right) $ with the property that $%
\lim_{k\rightarrow \infty }\delta G^{\left( k\right) }=E\sim \left(
I,0,0\right) $ in the uniform topology (i.e., in the operator norm). 
Then, for all $\Psi \in \mathfrak{h}\otimes \mathfrak{F}$,
we have 
\begin{eqnarray}
\lim_{k\rightarrow \infty }\Vert \big(U_{G^{(k)}}(t)-U_{G^{(k)}%
\vartriangleleft \delta G^{(k)}}(t)\big)\Psi \Vert =0. 
\label{eq:thm_limit}
\end{eqnarray}
\end{theorem}

Note that Theorem \ref{thm:main_divergent} leaves open the possibility that the
sequence of unperturbed systems, $(G^{(k)})_{k}$, need not have a well
defined limit. This is intentional, and will be exploited in applications.

In the Theorem we took the series product perturbation on the right. This
turns out to be more natural. We could also consider left series product
perturbations:\ $\widetilde{G}^{\left( k\right) }=\delta G^{\left( k\right)
}\vartriangleleft G^{\left( k\right) }$, however, we now need the following
expression
\begin{eqnarray}
&&\delta K^{\left( k\right) }+L^{\left( k\right) \ast }\left( e^{i\delta
\Theta ^{\left( k\right) }}-1\right) L^{\left( k\right) }+\left( \delta
L^{\left( k\right) }\right) ^{\ast }e^{i\delta \Theta ^{\left( k\right)
}}L^{\left( k\right) }-L^{\left( k\right) \ast }\delta L^{\left( k\right) }
\nonumber \\
&&+\alpha ^{\ast }S^{\left( k\right) \ast }\left[ \left( e^{i\delta \Theta
^{\left( k\right) }}-1\right) L^{\left( k\right) }+\delta L^{\left( k\right)
}\right] \nonumber \\
&& -\left[ L^{\left( k\right) \ast }\left( e^{i\delta \Theta ^{\left(
k\right) }}-1\right) +\delta L^{\left( k\right) }e^{i\delta \Theta ^{\left(
k\right) }}\right] S^{\left( k\right) }\alpha  \nonumber\\
&&+\alpha ^{\ast }S^{\left( k\right) \ast }\left( e^{i\delta \Theta
^{\left( k\right) }}-1\right) S^{\left( k\right) }\alpha ,
\label{eq:left_condition}
\end{eqnarray}
to vanish as $k\rightarrow \infty $ in the uniform topology (i.e., the operator norm) for each $\alpha
\in \mathbb{C}$. Note that there are terms such as $L^{(k)\ast } \, \delta L^{(k)}$ appearing in (\ref{eq:left_condition}) which
needs to be examined in detail since they may have competing scalings: the divergence of $L^{(k)}$ may compensate for the smallness of $\delta L^{(k)}$ for large $k$. This will then lead to the vanishing of (\ref{eq:left_condition}) becoming a nontrivial condition to be checked. For this reason, we prefer to work with the right perturbations.
In the special case where we consider perturbations about a fixed model $G^{(k)} =G$, then the condition $\delta G^{(k)} \to E \sim (I,0,0)$ uniformly suffices to ensure that (\ref{eq:left_condition}) vanishes. We summarize this in the following corollary.

\begin{corollary}
\label{cor:main}
Let $G$ be a fixed model with the system Hilbert
space, $\mathfrak{h}$, and multiplicity space, $\mathfrak{K}$. Let $(\delta
G^{(k)})_{k}$ be a sequence of perturbations $\delta G^{(k)}$ with the property that $%
\lim_{k\rightarrow \infty }\delta G^{\left( k\right) }=E\sim \left(
I,0,0\right) $ uniformly. 
Then, for all $\Psi \in \mathfrak{h}\otimes \mathfrak{F}$,
we have 
\begin{eqnarray}
\lim_{k\rightarrow \infty }\Vert \big(U_{G}(t)-U_{\widetilde{G}^{(k)}}(t)\big)\Psi \Vert =0, 
\label{eq:cor_limit}
\end{eqnarray}
where $\widetilde{G}^{(k)}$ is the left/right series product perturbation of $G$ by $\delta G^{(k)}$.
\end{corollary}

We note that a result related to the corollary has been obtained by Lindsay and Wills \cite
{LW07} for the case of perturbations about a fixed model. Since they worked in the more general framework of
contraction co-cycles (instead of unitaries) their results entail a weaker
mode of convergence, namely in the weak operator topology.

\bigskip

Theorem \ref{thm:main_divergent} motivates the following definition.

\begin{definition}
Let $(G^{(k)})_{k}$ and $(\widetilde{G}^{(k)})_{k}$ be sequences of models with the same system Hilbert
space $\mathfrak{h}$ and multiplicity space $\mathfrak{K}$. The sequences are asymptotically equivalent if there exists a sequence $(\delta G^{(k)})_{k}$ converging to $E$ such that $\widetilde{G}^{(k)} = G^{(k)} \vartriangleleft \delta G^{(k)}$.
\end{definition}

Evidently, asymptotically equivalent sequences of models satisfy equation (\ref{eq:thm_limit}).

\section{Proof of the main result}
\label{sec:proof}
The displacement operator $W_{\alpha }(t)=e^{\sum_{j}[\alpha _{j}A_{j}(t)^\ast  -\alpha _{j}^{\ast  
}A_{j}(t)]}$, for $\alpha \in \mathbb{C}^{n}$, is the solution to the QSDE 
\begin{eqnarray}
dW_{\alpha }\left( t\right) =\left\{ \alpha _{j}dA_{j}\left( t\right) ^{\ast  
}-\alpha _{j}^\ast  dA_{j}\left( t\right) -\frac{1}{2}\left| \alpha
_{j}\right| ^{2}dt\right\} W_{\alpha }\left( t\right) , 
\end{eqnarray}
equal to the identity at time zero.
We may think of $W_{\alpha }\left( t\right) $ as the unitary obtained from the
QSDE with the SLH coefficients $(I_{n},\alpha ,0)$. For a generator $G\sim \left[ S,L,K\right] $, we may perturb it with the generator of the displacement by means of the series product to get 
\begin{eqnarray}
G(\alpha )=G\vartriangleleft \left[ I_{n},\alpha ,-\frac{1}{2}\left| \alpha
\right| ^{2}\right] \sim \left[ S,L+S\alpha ,%
K-\frac{1}{2}\left| \alpha \right| ^{2}-L^\ast  S%
\alpha \right] . 
\label{eq:displaced}
\end{eqnarray}
Likewise, for a perturbed model $\widetilde{G}$ we may form $\widetilde{G}%
(\alpha )=\widetilde{G}\vartriangleleft \left[ I_{n},\alpha ,-\frac{1}{2}%
\left| \alpha \right| ^{2}\right] $.

\begin{lemma}
\label{lem:norm_continuity} Let $T_t $ be defined by the partial
trace over the Fock space as follows 
\begin{eqnarray}
T_t (X) \triangleq \text{tr}_{\mathfrak{F}} \left\{ I_{%
\mathfrak{h}} \otimes | \Omega \rangle \langle \Omega | \, U_{\widetilde{G}
(\alpha )} (t)^\ast   (X \otimes I_{\mathfrak{F}} ) U_{G(\alpha )}
(t) \right\} ,
\end{eqnarray}
where $\Omega$ is the Fock vacuum vector. Then $\left( T_t \right)_{t \ge 0 }$ is a norm continuous
one-parameter semigroup.
\end{lemma}

\begin{proof}
The semigroup property of $T_{t}$ follows immediately from the
cocycle property (w.r.t. the shift) of $U_{\widetilde{G}(\alpha )}(t)$ and $%
U_{G(\alpha )}(t)$. Since the conditional expectation given by the partial
trace is norm-contractive and $U_{\widetilde{G}(\alpha )}(t)$ and $%
U_{G(\alpha )}(t)$ are unitary, we have 
\begin{eqnarray}
\left\| T_{t}(X)\right\| \leq \left\| U_{\widetilde{G}(\alpha
)}(t)XU_{G(\alpha )}(t)\right\| \leq \left\| U_{\widetilde{G}(\alpha
)}(t)\right\| \left\| X\right\| \left\| U_{G(\alpha )}(t)\right\| \leq \Vert
X\Vert , 
\end{eqnarray}
and hence $T_{t}$ is norm-contractive. Note that due to the
boundedness of all coefficients in the QSDEs for $U_{\widetilde{G}(\alpha
)}(t)$ and $U_{G(\alpha )}(t)$, it follows immediately that the generator of 
$T_{t}$ is bounded. This means that $T_{t}$ is
norm-continuous.
\end{proof}

\begin{lemma}
\label{lem:dt_form} Let $G_{i}\sim \left[ S_{i},L_{i},%
K_{i}\right] $ then 
\begin{eqnarray*}
&&d\left( U_{G_{2}}\left( t\right) ^\ast  (X\otimes I_{\mathfrak{F}%
})U_{G_{1}}(t)\right) \nonumber \\
&=& U_{G_{2}}\left( t\right) ^\ast  \Big[ (K%
_{2}^\ast  X+XK_{1}+L_{2}^\ast  X L _{1})\otimes
I_{\mathfrak{F}}\Big] U_{G_{1}}(t)\,dt +\cdots
\end{eqnarray*}
where the remaining terms are proportional to $d\Lambda
_{ij},dA_{i},dA_{j}^\ast  $.
\end{lemma}

This follows directly from the quantum It\={o} calculus.

\begin{proposition}
\label{prop:dt_form} 
For a given model $G$ and a right series product perturbation $\widetilde{G}=G\vartriangleleft
\delta G$, where $\delta G \sim (e^{i \delta \Theta} , \delta L ,\delta H)$, 
we consider the displaced models $G\left( \alpha \right) $ and $%
\widetilde{G}\left( \alpha \right) $ as in (\ref{eq:displaced}).
\begin{eqnarray*}
&&d\left\{ U_{\widetilde{G}(\alpha )}(t)^\ast  (X\otimes I_{%
\mathfrak{F}})  U_{G(\alpha )}(t)  \right\} \\
&=&U_{\widetilde{G}(\alpha )}(t)^\ast   \Big[ (K^\ast  +\delta K%
^\ast  -\delta L^\ast  S^\ast  L-\alpha ^{\ast  
}e^{-i\delta \Theta }S^\ast  L-\alpha ^{\ast  
}e^{-i\delta \Theta }\delta L)X \\
&&+X\left( K-L^\ast  S\alpha \right) \\
&&+\left( L^\ast  +\delta L^\ast  S^\ast  +\alpha ^\ast  e^{-i\delta \Theta }S^\ast  \right) 
XL+S\alpha ) \Big]\otimes I_{\mathfrak{F}} \, U_{G(\alpha )}(t)  \,dt+%
\cdots
\end{eqnarray*}
where $\delta K = -i \delta H - \frac{1}{2} ( \delta L)^\ast \delta L$, and again the remaining terms are proportional to $d\Lambda_{ij},dA_{i},dA_{j}^\ast  $.
\end{proposition}

\begin{proof}
We have already $G(\alpha )\sim \left[ S,L+S%
\alpha ,K-\frac{1}{2}\left| \alpha \right| ^{2}-L^\ast  %
S\alpha \right] $ while 
\begin{eqnarray*}
\widetilde{G}(\alpha ) &=&\big[ Se^{i\delta \Theta },%
L+S\delta L,K+\delta K-L^\ast  S\delta L\big] \vartriangleleft \left[
I,\alpha ,-\frac{1}{2}\left| \alpha \right| ^{2}\right]  \nonumber \\
&=& \bigg[Se^{i\delta \Theta },L+S\delta 
L+Se^{i\delta \Theta }\alpha , \\
&&K+\delta K-L^\ast   S\delta L- L^\ast  Se^{i\delta \Theta }\alpha -\left(
\delta L\right) ^\ast  e^{i\delta \Theta }\alpha -\frac{1%
}{2}\left| \alpha \right| ^{2} \bigg] .
\end{eqnarray*}
The result now follows from Lemma \ref{lem:dt_form}.
\end{proof}

\begin{lemma}
\label{lem:final} For the models $G\left( \alpha \right) $ and $\widetilde{G}%
\left( \alpha \right) $ with $\widetilde{G}=G\vartriangleleft \delta G$, we
have 
\begin{eqnarray*}
&&d \bigg\{  U_{ \widetilde{G}(\alpha )}  (t)^\ast  U_{G(\alpha )}(t)\bigg\} = 
 U_{\widetilde{G}(\alpha )}(t)^\ast   \Big[ \delta K^\ast  \\
&&-\alpha ^\ast   e^{-i\delta \Theta } \delta L+(\delta L )^\ast  \alpha +\alpha ^\ast  \left( e^{-i\delta \Theta }-1\right) \alpha \Big] U_{G(\alpha )} (t) \bigg\} \, dt + \cdots
\end{eqnarray*}
where again the remaining terms are proportional to $d\Lambda
_{ij},dA_{i},dA_{j}^\ast  $.
\end{lemma}

This follows immediately from Proposition \ref{prop:dt_form} and the
identity (\ref{eq:K+Kdag}). Note that, specifically
\begin{eqnarray}
\delta K=-i\delta H- \frac{1}{2} \left( \delta L \right) ^\ast   \delta L.
\end{eqnarray}
The steps in this section may be retraced for the left series product perturbation leading to the more involved term corresponding to (\ref{eq:left_condition}).

We now have the following immediate corollary to Lemma \ref{lem:final} .

\begin{corollary}
\label{cor:Lto0} Let $\mathscr{L}$ be the generator of the
semigroup $T_{t}$. Then if the perturbation becomes trivial,
that is $\delta G\rightarrow E$ uniformly, then $\mathscr{L}(I)\rightarrow
0$ uniformly. (This means that $\mathscr{L}$ vanishes on the Banach algebra $\mathcal{B}_0 \equiv \mathbb{C} \, I$.)
\end{corollary}

\bigskip

We now turn to the proof of Theorem \ref{thm:main_divergent}. We rely heavily on the Trotter-Kato
theorem \cite{Trotter,Kato}. We have taken the formulation of the
Trotter-Kato theorem from \cite[Thm 3.17, page 80]{Dav80}.

\begin{theorem}[Trotter-Kato Theorem]
\label{thm:Trotter-Kato} Let $\mathcal{B}$ be a Banach space and let $%
\mathcal{B}_0$ be a closed subspace of $\mathcal{B}$. For each $n \ge 0$,
let $T_t^{(n)}$ be a strongly continuous one-parameter contraction semigroup
on $\mathcal{B}$ with generator $\mathscr{L}^{(n)}$. Moreover, let $T_t$ be
a strongly continuous one-parameter contraction semigroup on $\mathcal{B}_0$
with generator $\mathscr{L}$. Let $\mathcal{D}$ be a core for $\mathscr{L}$.
The following conditions are equivalent:

\begin{enumerate}
\item  For all $X \in \mathcal{D}$ there exist $X^{(n)} \in \mbox{Dom}\left(%
\mathscr{L}^{(n)}\right)$ such that 
\begin{eqnarray*}
\lim_{n \to \infty} X^{(n)} = X, \qquad \lim_{n \to \infty} \mathscr{L}%
^{(n)}\left(X^{(n)}\right) = \mathscr{L}(X).
\end{eqnarray*}

\item  For all $0 \le s < \infty$ and all $X \in \mathcal{B}_0$ 
\begin{eqnarray*}
\lim_{n \to \infty} \left\{\sup_{0 \le t \le s} \left\|T_t^{(n)}(X) -
T_t(X)\right\|\right\} = 0.
\end{eqnarray*}
\end{enumerate}
\end{theorem}

Let $(G^{(k)})_k$ be a sequence of models with the same system Hilbert
space, $\mathfrak{h}$, and multiplicity space, $\mathfrak{K}= \mathbb{C}^n$,
along with a sequence of perturbations, $(\delta G^{(k)})_k$. We take $%
(\widetilde G ^{(k)})_k$ to be the sequence of models obtained by right
series perturbation: $\widetilde{G}^{(k)}= G^{(k)} \vartriangleleft
\delta G^{(k)}$.

\bigskip

\begin{proof}
\textbf{[of Theorem \ref{thm:main_divergent}]} Let $T_{t}^{(k,\alpha )}$ be defined by 
\begin{eqnarray}
T_{t}^{(k,\alpha )}(X)\triangleq \text{tr}_{\mathfrak{F}%
}\left\{ I_{\mathfrak{h}}\otimes |\Omega \rangle \langle \Omega |\,U_{%
\widetilde{G}^{(k)}(\alpha )}(t)^\ast  (X\otimes I_{\mathfrak{F}%
})U_{G^{(k)}(\alpha )}(t)\right\} .
\end{eqnarray}
From Lemma \ref{lem:norm_continuity} we have that $T_{t}^{(k,\alpha )}$ is a norm-continuous one-parameter semigroup and so has a generator $\mathscr{L}^{(k,\alpha )}$. If we furthermore assume that the $\delta
G^{(k)}\rightarrow E$ then, by Corollary \ref{cor:Lto0}, we have that $\mathscr{L}^{(k,\alpha )}(I)\rightarrow 0$ as $k\rightarrow
\infty $ for all $\alpha \in \mathbb{C}$.

Our goal is to show that $\lim_{k\rightarrow \infty }\Vert
(U_{G^{(k)}}(t)-U_{\widetilde{G}^{(k)}}(t))\Psi \Vert =0$ for arbitrary
state $\Psi $, however, it is clearly sufficient to establish this for all
states of the $\Psi =v\otimes \exp (f)$ where $f$ belongs to some dense
subset of $\mathbb{C}^{n}\otimes L^{2}[0,\infty )$. We note that 
\begin{eqnarray}
&&\Vert (U_{G^{(k)}}(t)-U_{\widetilde{G}^{(k)}}(t))\Psi \Vert =  \nonumber \\
&& 2\Vert v\Vert
^{2}e^{\Vert f\Vert ^{2}} -2\text{Re}\bigg\langle v\otimes \exp (f), \, 
U_{\widetilde{G}^{(k)}}(t)^\ast  U_{G^{(k)}}(t)\, v\otimes \exp (f)\bigg\rangle .  \label{eq:id}
\end{eqnarray}
We now take the test functions $f$ to be simple. That is, there exist an $%
m<\infty $, times $0=t_{0},t_{1}<\cdots <t_{m}$ and constants $\alpha
_{1},\cdots ,\alpha _{m}\in \mathbb{C}^{n}$, such that 
\begin{eqnarray}
f(t)\equiv \sum_{j=1}^{m}\alpha _{j}\,1_{[t_{k-1},t_{k})}(t).
\label{eq:f_simple}
\end{eqnarray}
With $v\in \mathfrak{h}$ arbitrary, we have that (for $\alpha $ fixed and $%
s>t$) 
\begin{eqnarray}
&&\bigg\langle v\otimes \exp (\alpha 1_{[0,s)}),\, U_{\widetilde{G}^{(k)}}(t)^{\ast  
}U_{G^{(k)}}(t) \, v\otimes \exp (\alpha 1_{[0,s)})\bigg\rangle   \nonumber \\
&&\qquad =\bigg\langle v\otimes \Omega ,\, \bigg( U_{\widetilde{G}^{(k)}}(t)W_{\alpha
}(t)\bigg)^\ast   U_{G^{(k)}}W_{\alpha }(t)  \, v\otimes \Omega \bigg\rangle
e^{|\alpha |^{2}t}  \nonumber \\
&&\qquad =\bigg\langle v, \, T_{t}^{(k,\alpha )}(I)v\bigg\rangle e^{|\alpha
|^{2}t}.
\end{eqnarray}
More generally, for $f$ simple as in (\ref{eq:f_simple}), we have 
\begin{eqnarray}
&&\bigg\langle v\otimes \exp (f) , \, U_{\widetilde{G}^{(k)}}(t)^{\ast  
}U_{G^{(k)}}(t)\, v\otimes \exp (f))\bigg\rangle   \nonumber \\
&&\qquad =\bigg\langle v,\, T_{t_{1}}^{(k,\alpha _{1})}\circ \cdots
\circ T_{t_{m}}^{(k,\alpha _m)}(I)v\bigg\rangle \,e^{\sum_{j}|\alpha
_{j}|^{2}(t_{j}-t_{j-1})}.
\end{eqnarray} 
We now use the Trotter-Kato Theorem. In the statement of the Theorem \ref{thm:Trotter-Kato}, we take $\mathcal{B}_0 = \mathbb{C}I$ and $T_t = \text{id} $, the identity on $\mathcal{B}_0$. Note that the generator of $T_t$ is the zero operator. 
Furthermore we take $X^{(n)}$ and $X$ to be the identity $I$, while $\mathscr{L}^{(n)}$ is taken to be $\mathscr{L}^{(k,\alpha )}$ where we replace $k$ with $n$. We have the graph convergence (equivalent condition 1. of Theorem \ref{thm:Trotter-Kato}) if we can show the convergence $\mathscr{L}^{(k,\alpha
)}(I)\rightarrow 0$ as $k\rightarrow \infty $ for all $\alpha \in \mathbb{C}$,
but this is guaranteed by Corollary \ref{cor:Lto0}. This means that we may now invoke equivalent condition 2. of Theorem 
\ref{thm:Trotter-Kato}.

Therefore, we have 
\begin{eqnarray}
&&\lim_{k\rightarrow \infty }\bigg\langle v\otimes \exp (f), \, U_{\widetilde{G}%
^{(k)}}(t)^\ast  U_{G^{(k)}}(t) \, v\otimes \exp (f))\bigg\rangle   \nonumber \\
&&\qquad =\bigg\langle v\otimes \exp (f),\, v\otimes \exp (f)\bigg\rangle ,
\end{eqnarray}
and therefore, by (\ref{eq:id}), we have $\lim_{k\rightarrow \infty }\Vert
(U_{G^{(k)}}(t)-U_{\widetilde{G}^{(k)}}(t))\,v\otimes \exp (f)\Vert =0$. As
the set of simple functions forms a dense subset of $\mathbb{C}^{n}\otimes
L^{2}[0,\infty )$, and the exponential vectors form a total subset of the
Fock space \cite{Par92}, the result follows.
\end{proof}

\section{Applications}
\label{sec:applications}

\subsection{Virtual work}
The concept of virtual work for open quantum systems was introduced in \cite{AFG12} and the essential idea is that if the energy is perturbed by $\delta H$ and the coupling operators by $\delta L$ then the correct variation of energy should take the form 
\begin{eqnarray}
\Delta H   \triangleq \delta H+\text{Im}\left\{ \left( \delta L\right) ^\ast  L\right\} .
\label{eq:Delta_H}
\end{eqnarray}
Their argument is based on the fact that the Lindblad generator $\mathscr{L}$ given in (\ref{eq:Lindbladian}) is invariant with respect to the gauge transformations $L \mapsto RL$, for $R \in \mathbb{C}^{n \times n} $ unitary, and $ (L,H) \mapsto ( L+ \alpha , H + e + \text{Im} \{ \alpha^\ast L \} )$, for $e \in \mathbb{R}$ and $\alpha \in \mathbb{C}^n$. As such, $\Delta H$ in (\ref{eq:Delta_H}) is the variation that is invariant under this class of transformations.

In our formulation, we take a fixed model $G\sim \left( I,L,H\right) $ and consider a virtual perturbation 
$G\rightarrow G^{\prime }=\delta
G\vartriangleleft G\sim \left( I,L^{\prime },H^{\prime }\right) $
by $\delta G=\left(I ,\delta L,\delta H\right) $. Note that this is a left perturbation! 
The virtual work done by
the perturbation is then defined to be 
\begin{eqnarray}
\Delta H \triangleq H^{\prime }-H  \equiv \delta H+\text{Im}\left\{ \left( \delta L\right) ^\ast   L\right\} ,
\label{eq:virtual_work_SLH}
\end{eqnarray}
and this is identical with definition (\ref{eq:Delta_H}) in \cite{AFG12}.

(Note that the gauge invariance property can be restated in our language as the fact that $G$ and $(R, \alpha , e) \vartriangleleft G$ lead to the same Lindblad generator for all $R \in \mathbb{C}^{n \times n} $ unitary,  $\alpha \in \mathbb{C}^n$ and  $e \in \mathbb{R}$. Here, $(R, \alpha , e) $ may be thought of as an element of the Euclidean group over the Hilbert space $\mathfrak{K} = \mathbb{C}^n$.)

For instance, we may make a displacement by fixing a self-adjoint operator $%
F $, and performing the unitary rotations (here $  \phi$ is a real
parameter) to get 
\begin{eqnarray*}
 L (\phi )=e^{iF  \phi }Le^{ -iF  \phi },\quad 
 H (\phi )=e^{iF  \phi }He^{-iF \phi }.
\end{eqnarray*}
We consider a small parameter  $\delta \phi$ about $\phi =0$, then to lowest order we have changes $\delta H \equiv -i[H,F] \, \delta \phi$ and  $\delta L \equiv -i[L,F] \, \delta \phi$.
We therefore have
\begin{eqnarray}
G^\prime &\sim &  \bigg( I , L -i[L,F] \delta \phi +o (\delta \phi )  , \nonumber \\
&& \qquad H  -i[H,F]  \, \delta \phi +\text{Im}\left\{ \left( -i\left[
L,F\right] \right) ^\ast  L\right\} \delta \phi +o\left( \delta \phi \right) \bigg) .
\end{eqnarray}

Remarkably, the virtual work done is to first
order in $\delta \phi$ the Lindbladian of $F$: 
\begin{eqnarray}
\Delta H \equiv H - \mathscr{L}\left( F\right) \delta \phi +o\left( \delta \phi \right)
.
\label{eq:virtual_work_Lindbladian}
\end{eqnarray}
In particular, perturbations induced by observables $F$ do \textit{zero virtual work} if $F$ is a constant of the motion, see \cite{AFG12}.

Corollary \ref{cor:main} then states that the model generated by $G $ and its virtual perturbations agree in the limit where the perturbations
converge to the identity.

\subsection{Large $n$ limit of maximal squeezing}

We consider the use of large squeezing of a quantum input noise to replace the coupling with an effective term that picks out only
one quadrature of the noise, originally suggested at the level of the filter in \cite[Chapter 4]{Bouten_PhD}.

\bigskip

We recall the general set up of a squeezed quantum Wiener process, $B\left(
t\right) $, which satisfies an It\={o} table of the form 
\begin{eqnarray*}
\begin{tabular}{l|ll}
$\times $ & $dB^\ast  $ & $dB$ \\ \hline
$dB$ & $\left( n+1\right) dt$ & $mdt$ \\ 
$dB^\ast  $ & $m^\ast  dt$ & $ndt$%
\end{tabular}
\end{eqnarray*}
where $n\geq 0$ and $\left| m\right| ^{2}\leq \sqrt{n\left( n+1\right) }$.
(Note $n$ is to be distinguished from the previous notation for the multiplicity which is simply $\text{dim} \,\mathfrak{K} =1$ in this case.)
We will be interested in the maximally squeezed situation were 
\begin{eqnarray*}
m\equiv \sqrt{n\left( n+1\right) }e^{i\theta }.
\end{eqnarray*}
To obtain a unitary stochastic evolution on a system with Hilbert space $\mathfrak{h}$ driven by the squeezed noise, we consider a QSDE of the form 
\begin{eqnarray}
dU^{(n)}\left( t\right) =dG^{(n)}\left( t\right) \,U^{(n)}\left( t\right)
,\qquad U^{(n)}\left( 0\right) =I,
\label{eq:squeeze_B}
\end{eqnarray}
where $dG^{(n)} (t) = LdB(t)^\ast   + RdB(t) + K^{(n)} dt$ and we must have
the isometry condition $dG^{(n)}(t)+dG^{(n)   }(t)^\ast+dG^{(n)}\left( t\right)
^\ast  dG^{(n)}\left( t\right) =0$.

From the isometry condition, we must have $R= -L^\ast  $, and furthermore we find $%
K^{(n)}dt + K^{(n) \ast   }dt +\left[ LdB\left( t\right) ^\ast  -L^{\ast  
}dB\left( t\right) \right] ^\ast  \left[ LdB\left( t\right) ^{\ast  
}-L^\ast  dB\left( t\right) \right] =0$, which requires $K^{(n)} $ to be  
\begin{eqnarray*}
K^{(n)} =-\frac{1}{2}\left( n+1\right) L^\ast  L-\frac{1}{2}nL L^\ast   +%
\frac{1}{2}mL^{\ast   2}+\frac{1}{2}m^\ast  L^{2}-iH,
\end{eqnarray*}
for some $H$ self-adjoint on $\mathfrak{h}$.

It is possible for maximally squeezed noise to be represented in terms of
a standard Hudson-Parthasarathy Fock space process \cite{HP}, $A\left( t\right) $, by means of a Bogoliubov
transformation: 
\begin{eqnarray*}
B\left( t\right) =u\,A\left( t\right) +v\,A\left( t\right) ^\ast  .
\end{eqnarray*}

We follow \cite{HHKKR02}, and take the Bogoliubov transformation to be the one determined by
\begin{eqnarray*}
u^{(n)} &=& \frac{1}{\sqrt{\nu \left( n\right) }} \bigg[ n+1+\sqrt{n\left( n+1\right) }e^{i\theta } \bigg] ,\\
v^{(n)} &=& \frac{1}{\sqrt{\nu \left( n\right) }} \bigg[ n+\sqrt{n\left( n+1\right) }e^{i\theta } \bigg] ,
\end{eqnarray*}
where 
\begin{eqnarray*}
\nu \left( n\right) =2n+1+2\sqrt{n\left( n+1\right) }\cos \theta .
\end{eqnarray*}
Note that the QSDE (\ref{eq:squeeze_B}) may now be considered as being driven by the vacuum noise
process $A\left( t\right) $ with SLH terms $G \sim (I , L^{(n)} , H )$ where 
\begin{eqnarray*}
L^{\left( n\right) } &=&L\, u^{\left( n\right) \ast   }-L ^\ast   \,
v^{\left( n\right) } \\
&=&\frac{1}{\sqrt{\nu \left( n\right) }}\left[ \left( n+1+\sqrt{n\left(
n+1\right) }e^{-i\theta }\right) L -\left( n+\sqrt{n\left( n+1\right) }%
e^{i\theta }\right) L ^\ast  \right] .
\end{eqnarray*}

What is of interest here is the fact that $L^{\left( n\right) }$ is
approximately skew-adjoint, that is 
\begin{eqnarray}
L^{\left( n\right) }=F^{\left( n\right) }+\frac{1}{\sqrt{\nu \left( n\right) 
}}L ,
\end{eqnarray}
where 
\begin{eqnarray}
F^{\left( n\right) }=\frac{1}{\sqrt{\nu \left( n\right) }}\left[ \left( n+ 
\sqrt{n\left( n+1\right) }e^{-i\theta }\right) L -\text{h.c.}\right]
=-F^{\left( n\right) \ast   }.
\end{eqnarray}

We now wish to obtain an equivalent model $\widetilde{G}^{\left( n\right) }\sim \left( I,F^{\left( n\right) },\widetilde{H}^{\left( n\right) }\right) $ where we replace $L^{(n)}$ with its approximation $F^{(n)}$. We allow that the Hamiltonian term will need to shift to some new form $\tilde{H}^{(n)}$. This is achieved by taking $\widetilde{G}^{\left( n\right) }=G^{\left( n\right)
}\vartriangleleft \delta G^{\left( n\right) }$ with
\begin{eqnarray*}
\delta G^{\left( n\right) }\sim \left( I,-\frac{1}{\sqrt{\nu \left( n\right) 
}}L ,0\right) .
\end{eqnarray*}
This subtracts the unwanted part from the coupling $L^{(n)}$ and introduces by the series product the new Hamiltonian 
$\widetilde{H}^{\left( n\right) }=H+H_{n} $ where 
\begin{eqnarray*}
H_{n} &=& \text{Im}\left\{ L^{\left( n\right) \ast   } \delta L^{(n)}  \right\} 
= \frac{1}{\sqrt{\nu \left( n\right) }} \text{Im}\left\{ L^\ast   F^{\left( n\right) }\right\} \\
&=& -\frac{1}{2i}\left[ \frac{n+\sqrt{n\left( n+1\right) }e^{-i\theta }}{\nu
\left( n\right) }L^2 -\text{h.c.}\right] -\frac{\sqrt{n(n+1)}}{\nu \left(
n\right) }\sin \theta L ^\ast   L .
\end{eqnarray*}

Here the limit models are divergent as $F^{(n)} \asymp \sqrt{n} \left[ \frac{1+e^{-i\theta }}{2(1+\cos \theta ) }L -%
\text{h.c.}\right] $,
using the fact that $\nu (n) \asymp 2 (1+\cos \theta ) n $ for large $%
n$. We also have the limit 
\begin{eqnarray*}
H^\prime \triangleq \lim_{n \to n} H_{n} &=& -\frac{1}{2i}\left[ \frac{%
1+e^{-i\theta }}{2(1+\cos \theta ) }L ^2 -\text{h.c.}\right]
 -\frac{ \sin
\theta }{2(1+\cos \theta ) } L ^\ast  L .
\end{eqnarray*}

From Theorem \ref{thm:main_divergent} we conclude that the model generated by $\widetilde{G}^{(n)}$ is asymptotically equivalent to the original one $G^{(n)}$. 

The skew-adjointness of $F^{(n)}$ now implies that the asymptotically equivalent dynamics is given by
\begin{eqnarray}
d \widetilde{U}^{(n)} (t) = \bigg\{ F^{(n)} dQ (t) -\bigg( \frac{1}{2}
F^{(n)\ast  }F^{(n)} +i H +iH_n \bigg) dt \bigg\} \widetilde{U}^{(n)} (t),
\end{eqnarray}
where, significantly, we now only encounter  the specific quadrature $Q(t) = A (t) + A(t)^\ast  
$. We therefore obtain a class of essentially commutative dilations \cite{KM87} as the equivalent asymptotic model.

In a homodyne measurement scheme we measure the output quadrature process $Y^{(n)}(t) = \widetilde{U}^{(n)} (t)^\ast   Q(t) \widetilde{U}^{(n)} (t)$. However, we now have $Y^{(n)} (t) \equiv Q(t)$ since $[ Q(t) , \widetilde{U}^{(n)} (t) ] =0$. In other words, we are measuring white noise and
not extracting any information about the system. On the other hand, the
knowledge of $Q(t)$ allows us to reconstruct the actual unitary evolution of
the system entirely.

\subsection{Faraday rotation}

We consider a model describing a cold gas of atoms driven by a $y$-polarized laser field: the atoms undergo a weak Faraday rotation scattering photons from the $y$-polarization channel into the $x$-polarization channel; the scattering can be enhanced by taking the laser field to be in a large amplitude coherent state. Previously, it has been suggested \cite{BSSM07} that, in an appropriate limit (described below), the channels decouple. We can now make this argument precise.

In detail, we have the following SLH system 
\begin{eqnarray*}
G \sim \left(
S=\left[ 
\begin{array}{cc}
\cos \left( \kappa F_{z}\right) & -\sin \left( \kappa F_{z}\right) \\ 
\sin \left( \kappa F_{z}\right) & \cos \left( \kappa F_{z}\right)
\end{array}
\right] ,L=\left[ 
\begin{array}{c}
0 \\ 
0
\end{array}
\right] ,H=0
\right) ,
\end{eqnarray*}
describing the Faraday interaction between two inputs (the two linearly
polarized modes, $A^{x}\left( t\right) $ and $A^{y}\left( t\right) $, of a
light field) with a cloud of cold atoms as described in \cite{BSSM07},
Section III. Here $F_{z}$ is the $z$ -component of the collective spin
operator for the cloud of atoms and $\kappa $ is a coupling constant, which
will be taken to be small.

We add a large constant displacement $\alpha $ to the $y$-polarization
field, $A^{y}\left( t\right) $ beforehand, so that the model becomes 
\begin{eqnarray*}
\left( S,0,0\right) \vartriangleleft \left( I_{2},\left[ 
\begin{array}{c}
0 \\ 
\alpha
\end{array}
\right] ,0\right) =\left( S,\left[ 
\begin{array}{c}
-\sin \left( \kappa F_{z}\right) \alpha \\ 
\cos \left( \kappa F_{z}\right) \alpha
\end{array}
\right] ,0\right) .
\end{eqnarray*}
We take the following scaling 
\begin{eqnarray*}
\kappa \hookrightarrow \frac{1}{k}\kappa ,\qquad \alpha \hookrightarrow
k\alpha ,
\end{eqnarray*}
so that the model becomes $G^{\left( k\right) }\sim \left( S^{(k)}, L^{(k)} , 0 \right)$ with
\begin{eqnarray*}
 S^{(k)} = \left[ 
\begin{array}{cc}
\cos \left( \frac{\kappa }{k}F_{z}\right) & -\sin \left( \frac{\kappa }{k}%
F_{z}\right) \\ 
\sin \left( \frac{\kappa }{k}F_{z}\right) & \cos \left( \frac{\kappa }{k}%
F_{z}\right)
\end{array}
\right] , \quad  L^{(k)} =\left[ 
\begin{array}{c}
-k\sin \left( \frac{\kappa }{k}F_{z}\right) \alpha \\ 
k\cos \left( \frac{\kappa }{k}F_{z}\right) \alpha
\end{array}
\right]   ,
\end{eqnarray*}
which is divergent as $k\rightarrow \infty $. We note that the scattering matrix $S^{(k)}$ converges to the identity, but the coupling terms are
\begin{eqnarray*}
L^{(k)} = 
\left[ 
\begin{array}{c}
-k\sin \left( \frac{\kappa }{k}F_{z}\right) \alpha \\ 
k\cos \left( \frac{\kappa }{k}F_{z}\right) \alpha
\end{array}
\right]
\asymp 
\left[ 
\begin{array}{c}
- \kappa F_{z}  \alpha \\ 
k  \alpha
\end{array}
\right] .
\end{eqnarray*}

We now consider the following perturbation 
\begin{eqnarray*}
\delta G^{\left( k\right) }\sim \left( S^{(k)\ast   }, S^{(k)\ast   }
\left[ 
\begin{array}{c}
k\sin \left( \frac{\kappa }{k}F_{z}\right) \alpha -\kappa  F_{z} \alpha \\ 
k \left( 1-\cos \left( \frac{\kappa }{k}F_{z}\right) \right) \alpha 
\end{array}
\right] ,0\right) .
\end{eqnarray*}

Note that $\delta G^{(k)}$ vanishes as $k \to \infty$ so Theorem \ref{thm:main_divergent} applies with the equivalent sequence of models 

\begin{eqnarray*}
\tilde{G}^{k}=G^{\left( k\right) }\vartriangleleft \delta G^{\left( k\right)
}\sim \left( \left[ 
\begin{array}{cc}
1 & 0 \\ 
0 & 1
\end{array}
\right] ,\left[ 
\begin{array}{c}
-\kappa F_{z} \alpha \\ 
k  \alpha
\end{array}
\right] ,0\right) .
\end{eqnarray*}
This justifies the use of the filter, equation (24) in \cite{BSSM07}, based
on polarimetry or homodyne detection of the $x$-polarization output quadrature.

\subsection{Local asymptotic normality}

Following \cite{CBG} we consider the following parameterized family of
systems $G\left( \theta \right) \sim \left( I,L\left( \theta \right)
,H\left( \theta \right) \right) $, with $\theta $ in some open subset of $ \mathbb{R}$. 
The idea is that the family may exhibit the property of local asymptotic normality about a parameter value $\theta_0$. Specifically, we study the large sample size $n$ limit in quantum statistics with an anticipated scaling $\theta = \theta_0 + \frac{v}{\sqrt{n}}$ for some factor $v$.

Using Taylor series expansions, we assume that we may write 
\begin{eqnarray*}
L\left( \theta_0 +\Delta \theta \right) &=&L\left( \theta _{0}\right)
+L^{\prime }\left( \theta _{0}\right) \Delta \theta +\frac{1}{2}L^{\prime
\prime }\left( \theta _{0}\right) \left( \Delta \theta \right)
^{2}+R_{L}\left( \Delta \theta \right) , \\
H\left( \theta_0 +\Delta \theta \right) &=&H\left( \theta _{0}\right)
+H^{\prime }\left( \theta _{0}\right) \Delta \theta +\frac{1}{2}H^{\prime
\prime }\left( \theta _{0}\right) \left( \Delta \theta \right)
^{2}+R_{H}\left( \Delta \theta \right),
\end{eqnarray*}
where the remainder terms satisfy $R_{L}\left( \Delta \theta \right)
=o\left( \left( \Delta \theta \right) ^{2}\right) ,R_{H}\left( \Delta \theta
\right) =o\left( \left( \Delta \theta \right) ^{2}\right) $.

In the following, we relabel $\sqrt{n}$ as $k$, and consider the following scaling 
\[
\Delta \theta =\frac{v}{k},\qquad t\hookrightarrow k^{2}t.
\]
This leads to 
\begin{eqnarray*}
G^{\left( k\right) }_v &=&\left( I,kL\left( \theta _{0}+\frac{v}{k}\right)
,k^{2}H\left( \theta _{0}+\frac{v}{k}\right) \right) \\
&=& \Bigg( I,kL\left( \theta _{0}\right) +L^{\prime }\left( \theta
_{0}\right) v+\frac{1}{2}L^{\prime \prime }\left( \theta _{0}\right) \frac{%
v^{2}}{k}+kR_{L}\left( \frac{v}{k}\right) , \\
&&k^{2}H\left( \theta _{0}\right) +kH^{\prime }\left( \theta _{0}\right) v+%
\frac{1}{2}H^{\prime \prime }\left( \theta _{0}\right)
v^{2}+k^{2}R_{L}\left( \frac{v}{k}\right) \Bigg) .
\end{eqnarray*}
We note that 
\[
\lim_{k\rightarrow \infty }k^{2}R_{L}\left( \frac{v}{k}\right) =0,\quad
\lim_{k\rightarrow \infty }k^{2}R_{H}\left( \frac{v}{k}\right) =0.
\]

We now wish to retain the following dominant terms to give the following
equivalent asymptotic sequence of models 
\begin{eqnarray*}
\widetilde{G}_{v}^{\left( k\right) } &=&\Bigg(I,kL\left( \theta _{0}\right)
+L^{\prime }\left( \theta _{0}\right) v, \\
&&k^{2}H\left( \theta _{0}\right) +kH^{\prime }\left( \theta _{0}\right) v+%
\frac{1}{2}H^{\prime \prime }\left( \theta _{0}\right) v^{2}+\tilde{H}%
^{\left( k\right) }\Bigg),
\end{eqnarray*}
where $\tilde{H}^{\left( k\right) }$ is an additional Hamiltonian correction
to be determined. We seek $\delta G^{\left( k\right) }$ such that $%
\widetilde{G}_{v}^{\left( k\right) }=G_{v}^{\left( k\right)
}\vartriangleleft \delta G_{v}^{\left( k\right) }$ and using the series product
inversion we have
\begin{eqnarray*}
\delta G_{v}^{\left( k\right) } &=& \left( G_{v}^{\left( k\right) }\right) ^{-1}\vartriangleleft 
\widetilde{G}_{v}^{\left( k\right)}  \\
&\sim &\Bigg(I,-\frac{1}{2}L^{\prime \prime }\left( \theta _{0}\right) \frac{%
v^{2}}{k}-kR_{L}\left( \frac{v}{k}\right) , \, +\tilde{H}^{\left( k\right) }-k^{2}R_{H}\left( \frac{v}{k}\right)  \\
&&-\text{Im}\left[ \left( \frac{1}{2}L^{\prime \prime }\left( \theta
_{0}\right) \frac{v^{2}}{k}+kR_{L}\left( \frac{v}{k}\right) \right) ^{\ast }%
\bigg(kL(\theta _{0})+vL^{\prime }\left( \theta _{0}\right) \bigg)\right] %
\Bigg),
\end{eqnarray*}
and for $\delta G^{\left( k\right) }$ to converge to the group identity $%
E\sim \left( I,0,0\right) $ we need $\tilde{H}^{\left( k\right) }-\frac{v^{2}%
}{2}$Im$\left[ L^{\prime \prime }\left( \theta _{0}\right) ^{\ast }L(\theta
_{0})\right] $ to converge to zero. 

Therefore, from Theorem \ref{thm:main_divergent}, we obtain the equivalent asymptotic model
\begin{eqnarray*}
\widetilde{G}_{v}^{\left( k\right) } &=&\Bigg(I,kL\left( \theta _{0}\right)
+L^{\prime }\left( \theta _{0}\right) v, \\
&&k^{2}H\left( \theta _{0}\right) +kH^{\prime }\left( \theta _{0}\right) v+%
\frac{1}{2}H^{\prime \prime }\left( \theta _{0}\right) v^{2}+\frac{v^{2}}{2}%
\text{Im}\left[ L^{\prime \prime }\left( \theta _{0}\right) ^{\ast }L(\theta
_{0})\right] \Bigg) .
\end{eqnarray*}

The expression, $\frac{v^{2}}{2}H^{\prime \prime }\left( \theta _{0}\right) +\frac{v^{2}}{2}%
\text{Im}\left[ L^{\prime \prime }\left( \theta _{0}\right) ^{\ast }L(\theta
_{0})\right]$, that now appears in the equivalent asymptotic model is identified with the phase term occurring in the local asymptotic normality result \cite{CBG}.

\section*{Acknowledgement}
The authors would like to thank Institut Henri Poincar\'{e} for the hospitality shown during the Trimester on Measurement and Control of Quantum Systems April-July 2018 where this work was carried out. We would like to thank Pierre Rouchon and Hendra Nurdin for stimulating discussions during this time.


\bigskip

\begin{thebibliography}{99}
\bibitem{HP}  R.L. Hudson and K.R. Parthasarathy. Quantum Ito's formula and
stochastic evolutions. Commun. Math. Phys., 93:301323, (1984)

\bibitem{Par92}  K.R. Parthasarathy, \emph{An Introduction to Quantum
Stochastic Calculus} (Birkhauser, 1992).

\bibitem{AFL90}
 L. Accardi, A. Frigerio,  Y.G. Lu, 
The weak coupling limit as a quantum functional central limit, 
Comm. Math. Phys. 131,537-570 (1990)

\bibitem{G05}
J.E. Gough,
Quantum Flows as Markovian Limit of Emission, Absorption and Scattering Interactions,
Commun. Math. Phys. 254: 489,  (2005)

\bibitem{Gardiner_Collett}  C.W. Gardiner and M.J. Collett. Input and output
in damped quantum systems: Quantum stochastic differential equations and the
master equation. Phys. Rev. A, \textbf{31}(6):37613774, (1985)

\bibitem{GarZol00}  C.W. Gardiner and P. Zoller. \emph{Quantum Noise}
(Springer Berlin, 2000).

\bibitem{Belavkin92}  V.P. Belavkin,
Quantum stochastic calculus and quantum nonlinear filtering,
Journal of Multivariate Analysis,
Volume 42, Issue 2, 171-201 (1992).

\bibitem{BvHJ07}
L. Bouten, R. Van Handel, and M.R. James,
An Introduction to Quantum Filtering,
SIAM J. Control Optim., 46(6), 2199-2241 (2007).

\bibitem{GvH07}
J.E. Gough, R. van Handel,
Singular Perturbation of Quantum Stochastic Differential Equations with Coupling Through an Oscillator Mode,
Journ. of Stat. Phys., Volume 127, Issue 3, pp 575-607 (2007).

\bibitem{BS08}
L. Bouten, A. Silberfarb, 
Adiabatic Elimination in Quantum Stochastic Models,
Commun. Math. Phys. 283: 491 (2008).

\bibitem{GouJam09a}  J. Gough, M.R. James, Quantum Feedback Networks:
Hamiltonian Formulation, Commun. Math. Phys. \textbf{287}, 1109 (2009).

\bibitem{GouJam09b}  J. Gough, M.R. James, 
The series product and its
application to feedforward and feedback networks, 
IEEE Trans. on Automatic
Control \textbf{54}, 2530 (2009).

\bibitem{LW07}  J.M. Lindsay and S.J. Wills. Quantum stochastic operator
cocycles via associated semigroups,
Math. Proc. Camb. Phil. Soc.,
142:535-556, 2007.

\bibitem{BSSM07}  L. Bouten, J.K. Stockton, G. Sarma, H. Mabuchi, Scattering
of polarized laser light by an atomic gas in free space: a QSDE approach,
Phys. Rev. A \textbf{75}, 052111 (2007)

\bibitem{Trotter}  H. Trotter. Approximations of semigroups of operators.
Pacific J. Math., 8:887-919, (1958).

\bibitem{Kato}  T. Kato. Remarks on pseudo-resolvents and infinitesimal
generators of semigroups. Proc. Japan. Acad., 35:467-468, (1959).

\bibitem{KM87}
B. K\"{u}mmerer, H. Maassen,
The essentially commutative dilations of dynamical semigroups on $M_n$,
Comm. Math. Phys.,
Volume 109, Number 1, 1-22 (1987).

\bibitem{Bouten_PhD}
L. Bouten, Filtering and Control in Quantum Optics, PhD Thesis,
University of Nijmegen, arxiv:quant-ph/0410080 (2004).

\bibitem{Dav80}  E.B. Davies, One-parameter Semigroups, Academic Press Inc,
London (1980)

\bibitem{AFG12}  J.E. Avron, M. Fraas, G.M. Graf, Adiabatic response for
Lindblad dynamics, Journal of Statistical Physics, 148(5):800-823, (2012)

\bibitem{HHKKR02}
J. Hellmich, R. Honegger, C. K\"{o}stler, B. K\"{u}mmerer and A. Rieckers,
Couplings to Classical and Non-Classical Squeezed White Noise as Stationary Markov Processes,
Pub. RIMS, Volume 38, Issue 1, (2002)


\bibitem{CBG}  C. Catana, L. Bouten and M. Guta, Fisher informations and
local asymptotic normality for continuous-time quantum Markov processes,
Journ. of Physics A: Mathematical and Theoretical, Volume 48, Number 36
(2015)


\end{thebibliography}
\end{document}